\documentclass[journal, onecolumn]{IEEEtran}
\usepackage{amsmath,amsfonts,amsthm}
\usepackage{algorithmic}
\usepackage{algorithm}
\usepackage{array}
\usepackage[caption=false,font=normalsize,labelfont=sf,textfont=sf]{subfig}
\usepackage{textcomp}
\usepackage{stfloats}
\usepackage{url}
\usepackage{verbatim}
\usepackage{graphicx}
\usepackage{cite}
\usepackage{makecell}
\usepackage{flushend}
\usepackage{mathrsfs}

\newtheorem{theorem}{Theorem}   
   
\newtheorem{proposition}[theorem]{Proposition}  
\newtheorem{corollary}[theorem]{Corollary}

\theoremstyle{definition}
\newtheorem{definition}{Definition}
\theoremstyle{remark}
\newtheorem{remark}{Remark}

\newcommand{\R}{\mathbb{R}}
\newcommand{\dd}{\mathrm{d}}
\newcommand{\renyi}{\mathcal{R}}
\renewcommand{\H}{\mathcal{H}}
\newcommand{\N}{\mathcal{N}}

\newcommand{\B}{\mathcal{B}}
\newcommand{\smep}{\mathcal{N}}
\newcommand{\sme}{\mathcal{S}}
\newcommand{\T}{\mathcal{T}}
\newcommand{\st}{\mathcal{P}}
\newcommand{\diver}{\operatorname{div}}
\newcommand{\E}{\mathcal{E}}
\newcommand{\I}{\mathcal{I}}
\newcommand{\J}{\mathcal{J}}

\newcommand{\Tr}{\operatorname{Tr}}
\newcommand{\Q}{\mathcal{Q}}

\begin{document}

\title{The concavity of generalized entropy powers}

\author{Mario Bukal\IEEEmembership{~}
        % <-this % stops a space
\thanks{The author is with the University of Zagreb Faculty of Electrical Engineering and Computing, Unska 3, 10000 Zagreb, Croatia. Email: mario.bukal@fer.hr}% <-this % stops a space
%\thanks{}
}

%\markboth{Manuscript submitted to IEEE Transactions on Information Theory}%
%{Shell \MakeLowercase{\textit{et al.}}: A Sample Article Using  IEEEtran.cls for IEEE Journals}


\maketitle

\begin{abstract}
In this note we introduce a new family of entropy powers which are related to generalized entropies, called Sharma-Mittal entropies, and we prove their concavity along diffusion processes generated by $L^2$-Wasserstein gradient flows of corresponding entropy functionals. 
This result extends the result of Savar\'e and Toscani on the concavity of R\'enyi entropy powers (\emph{IEEE Trans.~Inf.~Theory}, 2014) and reveals a connection to R\'enyi entropy power inequalities by Bobkov and Marsiglietti (\emph{IEEE Trans.~Inf.~Theory}, 2017).
\end{abstract}

\begin{IEEEkeywords}
entropy power, Sharma-Mittal entropy, con\-ca\-vi\-ty, $L^2$-Wasserstein gradient flow
\end{IEEEkeywords}

\section{Introduction}
\IEEEPARstart{E}{ntropy} power has been introduced by Shannon in his seminal paper \cite{Sha48}. Given a continuous random vector $X$ with values in $\R^d$, the entropy power $\N(X)$ is defined by 
\begin{equation}\label{def:sep}
    \N(X) = \exp\left(\frac{2}{d}\H(X)\right)
\end{equation}
where 
\begin{equation*}
    \H(X) = - \int_{\R^d}u(x)\log u(x) \dd x\,,
\end{equation*}
is the Shannon (also known as Boltzmann-Gibbs) entropy and $u$ is
the probability density of $X$ \cite{CoTh91}. The entropy power is a superadditive functional, i.e.~for any two independent random vectors $X$ and $Y$ it holds
\begin{equation}\label{ineq:theEPI}
    \N(X+Y)\geq \N(X) + \N(Y)\,,
\end{equation}
with equality if and only if $X$ and $Y$ are Gaussian random vectors with independent identically distributed components.
This is the famous entropy power inequality (EPI), which was partially proved already by Shannon in \cite{Sha48}, who used it to obtain a lower bound on the channel capacity, but the complete proof was given later by Stam \cite{Sta59}. The EPI and its refinements were subject of extensive research in information theory, to name only few \cite{Dem89, DCT91, RaSa16, Rio11,VeGu06}.
Particularly important case are Gaussian perturbations of a random vector $X$. Let $Z$ be distributed according to the standard Gaussian, and let us denote $X_t = X + \sqrt{t}Z$ for $t\geq0$, then the following refinement of the EPI has been proved by Costa \cite{Cos85}
\begin{equation*}
    \N(X_t) \geq (1-t)\N(X_0) + t\N(X_1)\,,\quad \forall t\in[0,1]\,.
\end{equation*}
The latter can be rephrased as the concavity of the entropy power along the stochastic process $(X_t)_{t>0}$, i.e.
\begin{equation*}
    \frac{\dd^2}{\dd t^2}\N(X_t) \leq 0\,.
\end{equation*}
Since $(X_t)_{t>0}$ is a diffusion process whose probability densities $u_t(x)$, $t>0$, are governed by the linear diffusion or heat equation
\begin{align}\label{eq:heat}
    \partial_t u_t = \Delta u_t\,,\quad (x,t)\in\R^d\times(0,+\infty)\,,
\end{align}
one can also say that $\N(X_t)$ is concave along solutions to \eqref{eq:heat}, where $\Delta$ denotes the Laplace operator in $\R^d$. A remarkable proof of the concavity of the EPI was given by Villani in \cite{Vil00}, and we will explore this idea to prove our main result in Section \ref{sec:proof}.

The entropy $\H(X)$ and the heat equation \eqref{eq:heat} are intimately related in a geometric sense. Namely, it has been proved in \cite{JKO98} that the heat equation constitutes a gradient flow of the entropy functional $\widetilde{\H}(u) = -\H(X)$ with respect to a transport distance called $L^2$-Wasserstein distance on the space of probability densities of finite second moment. In notation $\widetilde{\H}(u) = -\H(X)$, $u$ always denotes the density of the corresponding random vector $X$. Informally speaking, we can write equation \eqref{eq:heat} in its Wasserstein gradient flow form \cite{JKO98}
\begin{equation*}
    \partial_t u_t = \diver\left(u_t\nabla \frac{\delta \widetilde{\H}(u_t)}{\delta u_t}\right)\,,
    \quad \frac{\delta \widetilde{\H}(u_t)}{\delta u_t} = \log u_t\,,
\end{equation*}
where $\delta\widetilde{\H}(u)/\delta u$ denotes the variational derivative of the functional $\widetilde{\H}(u)$, $\diver$ denotes the divergence operator, and $\nabla$ the Euclidean gradient in $\R^d$. Here we will use this gradient flow structure in a formal way and interested reader is referred to \cite{AGS05} for details. 

The above described link between the entropy power $\N(X)$ and the heat equation \eqref{eq:heat} was a cornerstone for an extension of the entropy power for R\'enyi entropies proposed by Savar\'e and Toscani in \cite{SaTo14}.
Defining $p$-th R\'enyi entropy power (for $p>1-2/d$) as
\begin{equation}\label{def:REP_ST}
    \st_p(X) = \exp\left(\sigma_p\renyi_p(X)\right)\,,
\end{equation}
where $\sigma_p = 2/d + p-1$ and 
\begin{equation}\label{def:Rp}
    \renyi_p(X) = \frac{1}{1-p}\log\left(\int_{\R^d} u^p(x)\dd x\right)
\end{equation}
is the R\'enyi entropy of order $p>0$, $p\neq1$,
they proved the concavity of $\st_p(X_t)$, i.e.
\begin{equation*}
    \frac{\dd^2}{\dd t^2}\st(X_t) \leq 0\,,
\end{equation*}
along the diffusion process whose probability densities $u_t(x)$, $t>0$, solve the nonlinear diffusion equation
\begin{equation}\label{eq:diff_p}
    \partial_t u_t = \Delta u^p_t\,,\quad (x,t)\in\R^d\times(0,+\infty)\,.
\end{equation}
This result was fruitful for obtaining Gagliardo-Nirenberg type functional inequalities with sharp constants \cite{Tos19} and improved decay rates for convergence of solutions of \eqref{eq:diff_p} to the self-similar profile \cite{CaTo14,Tos14}. In fact the choice of $\sigma_p$ in \eqref{def:REP_ST} comes from the requirement that the functional
\begin{equation*}
    \Q_p(u_t) = \st_p(X_t)\frac{\dd}{\dd t}\renyi_p(X_t)\,,
\end{equation*}
with the time derivative along solutions to \eqref{eq:diff_p},
is invariant with respect to mass conservative dilations, i.e.~$\Q_p(\lambda^d u(\lambda x)) = \Q_p(u(x))$ holds for all $\lambda>0$. Dilation invariance jointly with the self-similarity structure of solutions to \eqref{eq:diff_p} \cite{Vaz07} provides the isoperimetric inequality \cite[cf.~Theorem 2]{SaTo14}, which is the basis of Gagliardo-Nirenberg type inequalities.

Equation \eqref{eq:diff_p} is a well known and studied partial differential equation with plethora of applications \cite{Vaz07}. It also has a geometric interpretation in terms of the $L^2$-Wasserstein distance. For $p>\max(1-1/d,d/(d+2))$, $p\neq1$, it constitutes the $L^2$-Wasserstein gradient flow of the entropy functional $\widetilde{\T}_p(u) = -\T_p(X)$, where 
\begin{equation*}
    \T_p(X) = \frac{1}{1-p}\left(\int_{\R^d}u^p(x)\dd x - 1\right)\,,\quad p\neq 1\,,
\end{equation*}
is the Tsallis entropy of order $p$ \cite{Ott01}. In the gradient flow form, equation \eqref{eq:diff_p} reads
\begin{equation*}
    \partial_t u_t = \diver\left(u_t\nabla \frac{\delta \widetilde{\T}_p(u_t)}{\delta u_t}\right)\,,
    \quad \frac{\delta \widetilde{\T}_p(u_t)}{\delta u_t} = \frac{p}{p-1} u_t^{p-1}\,.
\end{equation*}
Hence, the R\'enyi entropy power $\st_p(X)$ is related to both R\'enyi and Tsallis entropies. We will resolve this ambiguity by introducing generalized entropy powers which will enable us to see a broad picture.

On the other hand, Bobkov and Chistyakov introduced another version of R\'enyi entropy power \cite{BoCh15}, which is a straightforward extension of the Shannon's entropy power, 
\begin{align*}
    \B_p(X) = \exp\left(\frac{2}{d}\renyi_p(X)\right)\,.
\end{align*}
Factor $2/d$ in the exponent makes the functional $\B_p(X)$, likewise $\N(X)$, homogeneous of order two, i.e.~$\B_p(\lambda X) = \lambda^2\B_p(X)$ for all $\lambda\in\R$.
They proved the following entropy power inequality: for $p>1$ and $n\geq3$, let $X_1$, $X_2$, $\ldots\ $, $X_n$ be independent continuous random vectors in $\R^d$, then
\begin{align*}
    \B_p(X_1 + X_2 + \ldots + X_n)\geq \frac{1}{e}p^\frac{1}{p-1}\sum_{k=1}^n\B_p(X_k)\,.
\end{align*}
Contrary to \eqref{ineq:theEPI}, the latter does not hold for $n=2$. A counterexample can be found in \cite{BoCh15}. The case of $p\in(0,1)$ has been recently discussed in \cite{MaMe19}.
Furthermore, in \cite{BoMa17} Bobkov and Marsiglietti extended the R\'enyi entropy power inequality to the following form: given independent continuous random vectors $X$ and $Y$ in $\R^d$, then
\begin{equation}\label{gepi:Bobkov}
    \B_p^\alpha(X+Y) \geq \B_p^\alpha(X) + \B_p^\alpha(Y)
\end{equation}
for all $\alpha\geq(p+1)/2$ and $p>1$.
\begin{remark}
Observe that for $\alpha = d(p-1)/2 + 1$ the $\alpha$-power of the  functional $\B_p(X)$ coincides with the functional $\st_p(X)$, i.e.~according to \eqref{gepi:Bobkov}, for $p>1$ we have the EPI
\begin{align}\label{epi:ST}
    \st_p(X+Y) \geq \st_p(X) + \st_p(Y)\,.
\end{align}
\end{remark}

Inspired by \eqref{gepi:Bobkov} we introduce a two-parameter generali\-zation of the entropy power that we will call
\emph{Sharma-Mittal entropy power of order $(p,q)$} (or simply \emph{Sharma-Mittal entropy power}). 
Its relation with Sharma-Mittal entropies \cite{ShMi75} will be clarified below.
\begin{definition} Let $X$ be a continuous random vector in $\R^d$. For $p > 1 - 2/d$ and $q>0$ we define Sharma-Mittal entropy power of order $(p,q)$ as
\begin{equation}\label{def:smep}
    \smep_{p,q}(X) = \exp\left(\sigma_q\renyi_p(X)\right)\,,
\end{equation}
where $\sigma_q = 2/d + q - 1$, and $\renyi_p(X)$ is the R\'enyi entropy of order $p$.
\end{definition}
\noindent It is apparent from the definition that setting $q = 2(\alpha-1)/d + 1$, for $\alpha\geq (p+1)/2$ and $p>1$, inequality (\ref{gepi:Bobkov})
reads as the EPI for Sharma-Mittal entropy powers
\begin{equation}\label{epi:SM}
    \smep_{p,q}(X+Y) \geq \smep_{p,q}(X) + \smep_{p,q}(Y).
\end{equation}
%\begin{remark}
In particular, for $q=p$, the Sharma-Mittal entropy power $\smep_{p,q}(X)$ coincides with the functional $\st_p(X)$, and inequality \eqref{epi:SM} reduces to \eqref{epi:ST}. If $q=1$, then $\smep_{p,q}(X)$ coincides with the functional $\B_p(X)$, but in this case $q=1$ implies $\alpha=1$, which in further requires $p=1$, hence \eqref{epi:SM} actually reduces to \eqref{ineq:theEPI}.
%\end{remark}

While the entropy power inequality \eqref{epi:SM} is an immediate consequence of \eqref{gepi:Bobkov} and the definition of $\smep_{p,q}(X)$, our aim in studying this subject was to complement \eqref{epi:SM} with the concavity property of $\smep_{p,q}(X_t)$ . Thus, to extend the result of Savar\'e and Toscani to the wider class of functionals. 
In light of the above interpretation of the concavity of entropy powers along the gradient flows of respective entropy functionals, it is appealing to formally consider the $L^2$-Wasserstein gradient flow 
\begin{equation}\label{gf:sme}
    \partial_t u_t = \diver\left(u_t\nabla \frac{\delta \widetilde{\sme}_{p,q}(u_t)}{\delta u_t}\right)\,,
\end{equation}
of functional $\widetilde{\sme}_{p,q}(u) = -\sme_{p,q}(X)$, 
where
\begin{equation*}
    \sme_{p,q}(X) = \frac{1}{1-q}\left[\left(\int_{\R^d}u(x)^p \dd x\right)^{\frac{1-q}{1-p}} - 1 \right]\,,\quad p,q\neq 1\,,
\end{equation*}
is the Sharma-Mittal entropy of order $(p,q)$ \cite{ShMi75}. Since the variational derivative of $\widetilde{\sme}_{p,q}(u)$ equals
\begin{equation*}
     \frac{\delta \widetilde{\sme}_{p,q}(u)}{\delta u}
     = \frac{p}{p-1}\left(\int_{\R^d}u^p\dd x\right)^{\frac{p-q}{1-p}}u^{p-1}\,,
\end{equation*}
partial differential equation \eqref{gf:sme} is for smooth positive solutions equivalent to the following
nonlinear and non-local diffusion equation
\begin{equation}\label{diff_eq_pq}
    \partial_t u_t = \left(\int_{\R^d}u_t^p\dd x\right)^{\frac{p-q}{1-p}}\Delta u_t^p\,,\quad (x,t)\in\R^d\times(0,+\infty)\,.
\end{equation}
This equation appeared in the literature \cite{FrDa00} in studying related diffusion processes.
Now we can state our main result which provides a remarkable geometric relation between generalized entropies and their powers.
\begin{theorem}\label{tm:main}
Let $p > 1 - 2/d$ and $q>0$, and let $(X_t)_{t\geq0}$ be a diffusion process whose probability densities $u_t(x)$, $t>0$, are smooth, strictly positive and rapidly decaying solutions to equation (\ref{diff_eq_pq}), then
\begin{equation*}
    \frac{\dd^2}{\dd t^2}\smep_{p,q}(X_t) \leq 0\,,\quad t>0\,.
\end{equation*}
\end{theorem}
\noindent We close this introductory section by few important remarks. 

First, observe that for $p=q$, the Sharma-Mittal entropy $\sme_{p,q}(X)$ equals to the Tsallis entropy $\T_p(X)$, hence equation (\ref{diff_eq_pq}) reduces to the nonlinear diffusion equation \eqref{eq:diff_p}, which is the gradient flow of $\widetilde{\T}_p(u)$. Since the concavity of the R\'enyi entropy power $\st_p(X)$, as called in \cite{SaTo14}, holds along the gradient flow of $\widetilde{\T}_p(u)$, from this perspective it could also be called the Tsallis entropy power. 

Second, on the limit as $q\to1$ the Sharma-Mittal entropy $\sme_{p,q}(X)$ becomes the R\'enyi entropy $\renyi_p(X)$, and equation (\ref{diff_eq_pq}) reduces to
\begin{equation}\label{gf:renyi}
     \partial_t u_t = \left(\int_{\R^d}u_t^p\dd x\right)^{-1}\Delta u_t^p\,,\quad (x,t)\in\R^d\times(0,+\infty)\,.
\end{equation}
which is formally the $L^2$-Wasserstein gradient flow of the R\'enyi entropy fun\-cti\-onal $\widetilde{\renyi}_p(u) = -\renyi_p(X)$ \cite{CaTo14}. Thus, according to Theorem \ref{tm:main}, the R\'enyi entropy power $\B_p(X)$ is concave along
the gradient flow of $\widetilde{\renyi}_p(u)$.
%although it does not satisfy the entropy power inequality, i.e.~\eqref{gepi:Bobkov} does not hold for $\alpha = 1$.

Last, but not least, the choice of $\sigma_q = 2/d + q-1$ in the definition of $\smep_{p,q}(X)$ follows an analogous argument like the choice of $\sigma_p$ in \cite{SaTo14}, as discussed above. Namely, we require that the functional 
\begin{equation*}
    \Q_{p,q}(u_t) = \smep_{p,q}(X_t)\frac{\dd}{\dd t}\renyi_p(X_t)\,,
\end{equation*}
with the time derivative along solutions to \eqref{diff_eq_pq},
is invariant with respect to mass conservative dilations.

In Section \ref{sec:aux} we outline algebraic relation between ge\-ne\-ra\-li\-zed entropies and respective entropy powers, which complements the above stressed geometric relation. In addition, we provide sufficient conditions for the entropy functional $\widetilde{\sme}_{p,q}(u)$ being geodesically convex which, according to the theory developed in \cite{AGS05}, makes the gradient flow structure \eqref{gf:sme} rigorous. The proof of Theorem \ref{tm:main} is given in Section \ref{sec:proof}.

\section{Generalized entropies and respective entropy powers}\label{sec:aux}

The Shannon's idea of axiomatic foundation of the entropy \cite{Sha48} has become a fertile ground for its generalizations. 
In \cite{Ren60} R\'enyi proposed an alteration of the Fadeev's set of axioms for the discrete Shannon's entropy, which in the continuous setting leads to the following entropy of order $p\neq 1$:
\begin{equation}\label{def:Rp2}
    \renyi_p(X) = \frac{1}{1-p}\log\left(\int_{\R^d} u^p(x)\dd x\right).
\end{equation}
Observe that $\lim_{p\to 1}\renyi_p(X) = \H(X)$.
Among others, further generalization has been proposed by Sharma and Mittal \cite{ShMi75}. They introduce a two parameter entropy of order $(p,q)$ as 
\begin{equation*}
    \sme_{p,q}(X) = \frac{1}{1-q}\left[\left(\int_{\R^d}u(x)^p \dd x\right)^{\frac{1-q}{1-p}} - 1 \right]\,,\quad p,q\neq 1\,.
\end{equation*}
Years later, in the framework of non-extensive ther\-mo\-dy\-na\-mics, Tsallis proposed a new family of entropies \cite{Tsa88}
\begin{equation*}
    \T_p(X) = \frac{1}{1-p}\left(\int_{\R^d}u^p(x)\dd x - 1\right)\,,\quad p\neq 1\,.
\end{equation*}
Again observe that $\lim_{p\to 1}\T_p(X) = \H(X)$, and furthermore $\sme_{p,p}(X) = \T_p(X)$ for $p\neq 1$. In the context of Tsallis statistics \cite{Tsa88, PlPl95} it is customary to work with 
$q$-logarithm
\begin{equation*}
    \log_q(s) = \frac{1}{1-q}\left(s^{1-q} - 1\right)\,,\quad s>0\,,
\end{equation*}
where $q\neq 1$, and its inverse, $q$-exponential
\begin{equation*}
    \exp_q(s) = \max\left(1 + (1-q)s, 0\right)^{\frac{1}{1-q}}\,.
\end{equation*}
In this notation we can write
\begin{equation*}
    \sme_{p,q}(X) = \log_q\left(\left(\int_{\R^d}u(x)^p \dd x\right)^{\frac{1}{1-p}}\right),
\end{equation*}
which better reveals similarities with R\'enyi entropies. Taking the $q$-exponential of the latter equation we find
\begin{equation}\label{smep-sme}
   \exp_q\left(\sme_{p,q}(X)\right) = \left(\int_{\R^d}u(x)^p \dd x\right)^{\frac{1}{1-p}},
\end{equation}
and therefore, the Sharma-Mittal entropy power can be written as
\begin{equation}\label{tep-te}
    \smep_{p,q}(X) = \left(\exp_q\left(\sme_{p,q}(X)\right)\right)^{\sigma_q}\,,
\end{equation}
where $\sigma_q = 2/d + q - 1$. This equation gives the direct relation between the Sharma-Mittal entropy and its power.
In particular, for $q=p$ we can write
\begin{equation*}
    \st_{p}(X) = \left(\exp_p\left(\T_{p}(X)\right)\right)^{\sigma_p}\,,
\end{equation*}
which gives the direct relation between the R\'enyi entropy power $\st_p(X)$ and the Tsallis entropy. Both equations \eqref{smep-sme} and \eqref{tep-te} resemble the original definition of the Shannon entropy power \eqref{def:sep}, which can be recovered on the limit as $p,q\to 1$:
\begin{equation*}
    \lim_{p,q\to1}\smep_{p,q}(X) = \lim_{p\to 1}\st_p(X) = \N(X)\,.
\end{equation*}
On the other hand, the R\'enyi entropy power $\B_p(X)$ follows as
\begin{equation*}
    \B_p(X) = \lim_{q\to 1}\smep_{p,q}(X)\,.
\end{equation*}

Although we work only formally with gradient flows, let us briefly mention under which conditions the above gradient flow structure is rigorous.
It has been shown in \cite[Chapter 9]{AGS05} that the functional
\begin{align}\label{def:Ep}
    \E_p(u) = \int_{\R^d}e_p(u(x))\dd x\, \quad \text{ with }\ e_p(z) = \frac{1}{p-1}z^p\,,
\end{align}
is for $p \geq 1 - 1/d$, $p\neq 1$, geodesically convex on the space of probability measures of finite second moment $\mathscr{P}_2(\R^d)$. Roughly speaking this means that $\E_p(u)$ is convex along the geodesic curve $(u_t)_{0\leq t\leq 1}$ connecting any two measures $u_0, u_1\in \mathscr{P}_2(\R^d)$, i.e.
\begin{equation}\label{geo_conv}
    \E_p(u_t) \leq (1-t)\E_p(u_0) + t\E_p(u_1)\,, \quad \forall t\in[0,1]\,.
\end{equation} 
At the expense of rigor, but for the sake of simplicity of exposition, we denote both probability measures and their densities with respect to the Lebesgue measure simply by $u$. 

Now observe that for $p>1$ the Sharma-Mittal entropy functional can be written as
\begin{align*}
    \widetilde{\sme}_{p,q}(u) &= -\log_q\left(\left((p-1)\E_p(u)\right)^{\frac{1}{1-p}}\right) = s_{p,q}\left(\E_p(u)\right)\,,
\end{align*}
where 
\begin{equation*}
    s_{p,q}(z) = -\log_q\left(\left((p-1)z\right)^{\frac{1}{1-p}}\right).
\end{equation*}
Easy calculation gives that $s_{p,q}(z)$ is non-decreasing for $p>1$ and convex for $q\geq p > 1$. Hence, the composition with $\E_p(u)$ and \eqref{geo_conv} yield the geodesic convexity of $\widetilde{\sme}_{p,q}(u)$ when $q\geq p > 1$. Then, according to the theory developed in \cite{AGS05}, in this range of parameters the gradient flow structure \eqref{gf:sme} is well-posed. 

\section{Proof of Theorem \ref{tm:main}}\label{sec:proof}

In order to prove our main result, we closely follow the approach of Savar\'e and Toscani in \cite{SaTo14}. After introducing auxiliary functional $\E_p(u)$ in \eqref{def:Ep}, they also introduce a generalization of the Fisher information
\begin{equation}\label{def:Fp}
    \I_p(u) = \int_{\R^d}\frac{|\nabla u^p|}{u}\dd x
    = \int_{\R^d}u|\nabla e_p'(u)|^2\dd x\,,
\end{equation}
and the second-order functional
\begin{equation}\label{def:Jp}
    \J_p(u) = 2\int_{\R^d}u^p\left(|\nabla^2e_p'(u)|^2 + (p-1)(\Delta e_p'(u))^2\right)\dd x\,.
\end{equation}
The following proposition has been proved in \cite[Proposition 3]{SaTo14}.
\begin{proposition}
Let $u_t(x)$, $t>0$, be smooth, strictly positive and rapidly decaying probability densities solving the nonlinear diffusion equation \eqref{eq:diff_p}, then
\begin{align}\label{eq:dEpST}
    -\frac{\dd}{\dd t}\E_p(u_t) &= \I_p(u_t)\,,\quad t>0\,,\\
    -\frac{\dd}{\dd t}\I_p(u_t) &= \J_p(u_t)\,, \quad t>0\,.\label{eq:dIpST}
\end{align}
\end{proposition}
\noindent Identity \eqref{eq:dEpST} actually says that the generalized Fisher information $\I_p(u)$ equals to the production of the Tsallis entropy functional $\widetilde{\T}_p(u)$ along its own gradient flow. Equation \eqref{eq:dIpST} can be interpreted as the production of the Fisher information along the gradient flow of $\widetilde{\T}_p(u)$, which gives the second-order functional $\J_p(u)$. 

Using these results we can prove the following analogous statement.

\begin{corollary}\label{corr:main}
Let $u_t(x)$, $t>0$, be smooth, strictly positive and rapidly decaying probability densities solving the nonlinear diffusion equation \eqref{diff_eq_pq}, then
\begin{align}
    -\frac{\dd}{\dd t}\E_p(u_t) &= \label{eq:dEp} \left((p-1)\E_p(u_t)\right)^{\frac{p-q}{1-p}}\I_p(u_t)\,,\quad t>0\,,\\
    -\frac{\dd}{\dd t}\I_p(u_t) &= \label{eq:dIp} \left((p-1)\E_p(u_t)\right)^{\frac{p-q}{1-p}}\J_p(u_t)\,, \quad t>0\,.
\end{align}
\end{corollary}

\begin{proof}
Observe that for smooth and strictly positive solutions equation \eqref{diff_eq_pq} can be equivalently written as
\begin{equation*}
    \partial_tu_t = \left(\int_{\R^d}u_t^p\dd x\right)^{\frac{p-q}{1-p}}\nabla\cdot(u\nabla e_p'(u_t))\,.
\end{equation*}
Thus, assuming in addition rapid decay of solutions to \eqref{diff_eq_pq} we can freely integrate by parts and calculate:
\begin{align*}
    \frac{\dd}{\dd t}\E_p(u_t) &= \int_{\R^d}e_p'(u_t)\partial_tu_t\\
    &= \left(\int_{\R^d}u_t^p\dd x\right)^{\frac{p-q}{1-p}}\int_{\R^d}e_p'(u_t)\nabla\cdot\left(u\nabla e_p'(u_t)\right)\dd x\\
    &= -\left(\int_{\R^d}u_t^p\dd x\right)^{\frac{p-q}{1-p}}\int_{\R^d}u |\nabla e_p'(u_t)|^2\dd x\,.
\end{align*}
Identity \eqref{eq:dEp} then obviously follows from definitions \eqref{def:Ep}
and \eqref{def:Fp}. Identity \eqref{eq:dIp} follows in analogous straightforward way from \eqref{eq:dIpST}. 
\end{proof}

\noindent Recall the definition of the Sharma-Mittal entropy power \eqref{def:smep},
%$\smep_{p,q}(X)$ can be written in terms of functional $\E_p$ as
\begin{equation*}
    \smep_{p,q}(X) = \exp\left(\sigma_q\renyi_p(X)\right),
\end{equation*}
where $\sigma_q = 2/d + q-1$.
%To simplify the exposition of the proof, let us consider a functional of the form
%\begin{equation}
%    \M(X) = \exp(\sigma\F(u))\,,
%\end{equation}
%for some $\sigma > 0$. Then 
Taking a stochastic process $(X_t)_{t\geq0}$, we simply calculate
\begin{align*}
    \frac{\dd}{\dd t}\smep_{p,q}(X_t) &= \sigma_q\smep_{p,q}(X_t)\frac{\dd}{\dd t}\renyi_p(X_t)\,,\\
    \frac{\dd^2}{\dd t^2}\smep_{p,q}(X_t) &= \sigma_q\smep_{p,q}(X_t)\left(\sigma_q\left(\frac{\dd \renyi_p(X_t)}{\dd t}\right)^2%\right.\\
    %&\qquad\qquad \left. 
    + \frac{\dd^2\renyi_p(X_t)}{\dd t^2} \right).
\end{align*}
Therefore, $\smep_{p,q}(X_t)$ is concave, i.e.~$\dd^2\smep_{p,q}(X_t)/\dd t \leq 0$ if and only if 
\begin{equation}\label{ineq:ccc}
    -\frac{\dd^2\renyi_p(X_t)}{\dd t^2} \geq \sigma_q\left(\frac{\dd \renyi_p(X_t)}{\dd t}\right)^2\,.
\end{equation}
Writing the R\'enyi entropy $\renyi_p(X)$ in terms of the functional $\E_p(u)$ as 
\begin{equation}
    \renyi_p(X) = \frac{1}{1-p}\log\left((p-1)\E_p(u)\right),
\end{equation}
we further calculate
\begin{align*}
    \frac{\dd }{\dd t}\renyi_p(X_t) &= \frac{1}{1-p}\frac{\dfrac{\dd}{\dd t}\E_p(u_t)}{\E_p(u_t)}\,, \\
    \frac{\dd^2 }{\dd t^2}\renyi_p(X_t) &= \frac{1}{1-p}\left(\frac{\dfrac{\dd^2}{\dd t^2}\E_p(u_t)}{\E_p(u_t)} - \left(\frac{\dfrac{\dd}{\dd t}\E_p(u_t)}{\E_p(u_t)}\right)^2
    \right).
\end{align*}
If $(X_t)_{t\geq0}$ is a stochastic process whose density function $u_t(x)$ solves \eqref{diff_eq_pq}, then employing the identities from Corollary \ref{corr:main} it follows
\begin{align*}
    \frac{\dd }{\dd t}\renyi_p(X_t) &= \left((p-1)\E_p(u_t)\right)^{\frac{2p-q-1}{1-p}}\I_p(u_t)\,, \\
    \frac{\dd^2 }{\dd t^2}\renyi_p(X_t) &= \left((p-1)\E_p(u_t)\right)^{\frac{3p-2q-1}{1-p}}%\\
    %&\quad\quad \times
    \left(\frac{2p-q-1}{p-1}\frac{\I_p^2(u_t)}{\E_p(u_t)} - \J_p(u_t)
    \right).
\end{align*}
Concavity condition \eqref{ineq:ccc} then becomes equivalent to 
\begin{equation*}
    \J_p(u_t) \geq (\sigma_q + 2p-q-1)\frac{\I_p^2(u_t)}{(p-1)\E_p(u_t)}\,,
\end{equation*}
which can be further written as
\begin{equation}\label{ineq:JpIp2}
    \J_p(u_t)\left(\int_{\R^d}u_t^p\dd x\right) \geq 2\left(\frac{1}{d} + p - 1\right)\I_p^2(u_t)\,.
\end{equation}
The last inequality has been demonstrated in \cite[cf.~inequality (26)]{SaTo14}, hence, the proof of Theorem \ref{tm:main} is finished. 

However, we provide an alternative proof of \eqref{ineq:JpIp2} which takes the idea from the Villani's proof of the concavity of the Shannon's entropy power $\N(X_t)$ \cite{Vil00}.
Elementary trace inequality $(\Tr(A))^2/d \leq |A|^2$ for $A\in\R^{d\times d}$ gives us that
\begin{equation*}
    \frac{1}{d}(\Delta e_p'(u) + d\lambda)^2 \leq |\nabla^2e_p'(u) + \lambda I|^2
\end{equation*}
for arbitrary $\lambda\in\R$, where $I$ denotes the identity matrix $d\times d$.
Hence, for every $\lambda\in\R$ and $p > 1 - 1/d$ it holds
\begin{align*}
    0\leq \int_{\R^d}u^p\left(|\nabla^2e_p'(u) + \lambda I|^2 + (p-1)(\Delta e_p'(u) + d\lambda)^2\right)\dd x.
\end{align*}
Expanding the right hand side we obtain
\begin{align*}
   0&\leq \int_{\R^d}u^p\left(|\nabla^2e_p'(u)|^2 + 2\lambda\Delta e_p'(u) + d\lambda^2\right)\dd x \\
    &\quad +(p-1)\int_{\R^d}u^p\left((\Delta e_p'(u))^2 + 2d\lambda\Delta e_p'(u) + d^2\lambda^2\right)\dd x\,.
\end{align*}
Integrating by parts and using the fact that $\nabla u^p = u\nabla e'_p(u)$ we arrive to
\begin{align*}
    0&\leq \int_{\R^d}u^p|\nabla^2e_p'(u)|^2\dd x - 2\lambda \I_p(u) + d\lambda^2\int_{\R^d}u^p\dd x \\
    &\quad +(p-1)\int_{\R^d}u^p(\Delta e_p'(u))^2\dd x %\\
    %&\quad 
    - 2d\lambda(p-1)\I_p(u) + d^2\lambda^2(p-1)\int_{\R^d}u^p\dd x\,.
\end{align*}
Now taking 
\begin{align*}
    \lambda = \dfrac{\I_p(u)}{d\int_{\R^d}u^p\dd x}\,,
    %\mu = \dfrac{\I_p(u)}{\int_{\R^d}u^p\dd x}
\end{align*}
the inequality becomes
\begin{align*}
    0&\leq \int_{\R^d}u^p|\nabla^2e_p'(u)|^2\dd x - \frac{1}{d}\dfrac{\I_p^2(u)}{\int_{\R^d}u^p\dd x} %\\
    %&\quad 
    +(p-1)\int_{\R^d}u^p(\Delta e_p'(u))^2\dd x - (p-1)\dfrac{\I_p^2(u)}{\int_{\R^d}u^p\dd x}\,.
\end{align*}
Multiplying the latter by $\displaystyle 2\int_{\R^d}u^p\dd x$ and rearranging terms we find
\begin{align*}
    2\int_{\R^d}u^p\left(|\nabla^2e_p'(u)|^2 +(p-1)(\Delta e_p'(u))^2\right)\dd x\left( \int_{\R^d}u^p\dd x\right)%\\
    \geq 2\left(\frac{1}{d} + p - 1\right)\I_p^2(u)\,,
\end{align*}
which is exactly inequality \eqref{ineq:JpIp2}.

\section{Conclusion}

In this note we introduced generalized entropy powers, called Sharma-Mittal entropy powers, and complemented the result of Bobkov and Marsiglietti \cite{BoMa17} by proving the concavity of Sharma-Mittal entropy powers along the $L^2$-Wasserstein gradient flows of the corresponding Sharma-Mittal entropy functionals. Thus, we generalized the result of Savar\'e and Toscani \cite{SaTo14}, which could be from our perspective restated as "the concavity of the Tsallis entropy power". Since the R\'enyi entropy power proposed by Bobkov and Chistyakov \cite{BoCh15} is geometrically related to the R\'enyi entropy, in this way we dissolved the ambiguity of the R\'enyi entropy power in the literature. To conclude, the contribution of our result and its settlement in the literature is best seen from the following table.
\begin{table}[H]
    \begin{center}
    \begin{tabular}{|c|c|c|c|c|}\hline
        \thead{~} & $\N(X)$ & $\st_p(X)$ & $\B_p(X)$ & $\smep_{p,q}(X)$ \\ \hline
        \thead{EPI} & \makecell{Stam \\ \cite{Sta59}} & \makecell{Bobkov and\\ Marsiglietti\\ \cite{BoMa17}} & \makecell{Bobkov and \\ Chistyakov\\ \cite{BoCh15}} & \makecell{Bobkov and\\ Marsiglietti\\ \cite{BoMa17}}\\ \hline
        \thead{concavity} & \makecell{Costa \cite{Cos85},\\ Villani \cite{Vil00}} & \makecell{Savar\'e and \\ Toscani\\ \cite{SaTo14} }& \makecell{Theorem \ref{tm:main}, \\ $q=1$} & Theorem \ref{tm:main}\\ \hline 
    \end{tabular}
    \end{center}
    \caption{Contribution and settlement of our result in the literature}
    \label{tab:my_label}
\end{table}

\section*{Acknowledgments}
This work has been supported by the Croatian Science Foundation under project UIP-05-2017-7249 (MANDphy).

%\newpage

\begin{IEEEbiographynophoto}{Mario Bukal}
is an Associate Professor of Mathematics at the University of Zagreb Faculty of Electrical Engineering and Computing since 2016. He obtained his PhD in mathematics from the Vienna University of Technology in 2012. His research interest are in entropy methods for diffusion equations, scaling limits in continuum mechanics and information fusion.
\end{IEEEbiographynophoto}

\vfill


\begin{thebibliography}{1}
\bibliographystyle{IEEEtran}

\bibitem{AGS05} L.~Ambrosio, N.~Gigli and G.~Savar\'e. Gradient Flows
in Metric Spaces and in the Space of Probability Measures. Lectures in Mathematics, Birkh\"auser, Basel, 2005.

\bibitem{BoCh15} S.~G.~Bobkov and G.~P.~Chistyakov. Entropy power inequality for the R\'enyi entropy. {\em IEEE Trans.~Inform.~Theory} 61 (2015), no.~2, 708-714.

\bibitem{BoMa17} S.~G.~Bobkov and A.~Marsiglietti. Variants of the Entropy Power Inequality. {\em IEEE Trans.~Inform.~Theory} 63 (2017), no.~12, 7747-7752.

\bibitem{CaTo14} J.~A.~Carrillo and G.~Toscani. R\'enyi entropy and improved equilibration rates to self-similarity for nonlinear diffusion equations. {\em Nonlinearity} 27 (2014), 3159-3177.

\bibitem{Cos85} M.~Costa. A new entropy power inequality. 
{\em IEEE Trans.~Inf.~Theory} 31 (1985), no.~6, 751-760.

\bibitem{CoTh91} T.~M.~Cover and J.~A.~Thomas. Elements of information theory. Wiley Series in Telecommunications. A Wiley-Interscience Publication. John Wiley \& Sons, Inc., New York, 1991. xxiv+542 pp.

\bibitem{Dem89} A.~Dembo. A simple proof of the concavity of the entropy power with respect to the variance of additive normal noise. {\em IEEE Trans.~Inform.~Theory} 35 (1989), 887-888.

\bibitem{DCT91} A.~Dembo, T.~M.~Cover and J.~A.~Thomas. Information-theoretic inequalities. {\em IEEE Trans.~Inform.~Theory}, 37 (1991), no.~6, 1501-1518.

\bibitem{FrDa00} T.~D.~Frank and A.~Daffertshofer. Exact time-dependent solutions of the R\'enyi Fokker-Planck equation and the Fokker-Planck equations related to the entropies proposed by Sharma and Mittal. {\em Physica A} 285 (2000), 351-366.

\bibitem{JKO98} R.~Jordan, D.~Kinderlehrer and F.~Otto. The Variational Formulation of the Fokker-Planck Equation. {\em SIAM J.~Math.~Anal.} 29 (1998), 1-17.

\bibitem{MaMe19} A.~Marsiglietti and J.~Melbourne. On the entropy power inequality for the R\'enyi entropy of order $(0,1)$. 
{\em IEEE Trans.~Inform.~Theory}, 65 (2019), 1387-1396.

\bibitem{Ott01} F.~Otto. The geometry of dissipative evolution equations: the porous medium equation. {\em Commun.~Part.~Diff.~Eq.} 26 (2001), 101-174.

\bibitem{PlPl95} A.~R.~Plastino and A.~Plastino. Non-extensive statistical mechanics and generalized Fokker-Planck equation. {\em Physica A} 222 (1995), 347-354.

\bibitem{RaSa16} E.~Ram and I.~Sason. On R\'enyi Entropy Power Inequalities. {\em IEEE Trans.~Inform.~Theory} 62 (2016), no.~12, 6800-6815.

\bibitem{Ren60} A.~R\'enyi. On Measures of Entropy and Information. {\em Proc.~Fourth Berkeley Symp.~on Math.~Stat.~and Probability}, University of California Press, (1960), 547-561.

\bibitem{Rio11} O.~Rioul. Information theoretic proofs of entropy power inequalities. {\em IEEE Trans.~Inform.~Theory} 57 (2011), no.~1, 33-55.

\bibitem{SaTo14} G.~Savar\'e and G.~Toscani. The concavity of R\'enyi entropy power. {\em IEEE Trans.~Inform.~Theory} 60 (2014), no.~5, 2687-2693.

\bibitem{Sha48} C.~E.~Shannon. A mathematical theory of communication. {\em Bell System Tech.~J.} 27 (1948), 379-423, 623-656.

\bibitem{ShMi75} B.~D.~Sharma and  D.~P.~Mittal. New Non-additive Measures of Entropy for Discrete Probability Distributions. {\em Journal of Mathematical Sciences} 10 (1975), 28-40.

\bibitem{Sta59} A.~J.~Stam. Some inequalities satisfied by the quantities of information of Fisher and Shannon. {\em Information and Control} 2 (1959), 101-112.

\bibitem{Tos14} G.~Toscani. R\'enyi entropies and nonlinear diffusion equations. {\em Acta.~Appl.~Math.} 132 (2014), 595-604.

\bibitem{Tos19} G.~Toscani. The information-theoretic meaning of Gagliardo-Nirenberg type inequalities. {\em Rend.~Lincei Mat.~Appl.} 30 (2019), 237-253.

\bibitem{Tsa88} C.~Tsallis. Possible generalization of Boltzmann–Gibs statistics. {\em J.~Stat.~Phys.} 52 (1988), 479-487.

%\bibitem{UTS08} S.~Umarov, C.~Tsallis, S.~Steinberg. On a q-central limit theorem consistent with nonextensive statistical mechanics. {\em Milan J.~Math.} 76 (2008), 307-328.

\bibitem{Vaz07} J.~L.~Vazquez. The Porous Medium Equation: Mathematical Theory. Oxford, UK: Oxford University Press, 2007.

\bibitem{VeGu06} S.~Verd\'u and D.~Guo. A simple proof of the entropy-power inequality. {\em IEEE Trans.~Inform.~Theory} 52 (2006), no.~5, 2165-2166.

\bibitem{Vil00} C.~Villani. A short proof of the ”concavity of entropy power”. {\em IEEE Trans.~Inform.~Theory} 46 (2000), no.~4, 1695-1696.


\end{thebibliography}
\end{document}